\documentclass[article]{IEEEtran}

\usepackage{graphicx,colordvi,psfrag,bbm}
\usepackage{amsmath,amssymb}
\usepackage{epstopdf}
\usepackage[caption=false]{subfig}
\usepackage{epsfig,cite}
\usepackage{calc,pstricks, pgf, xcolor}
\usepackage{nicefrac}
\usepackage{enumerate}

\newtheorem{theorem}{Theorem}

\newtheorem{remark}{Remark}

\newtheorem{lemma}{Lemma}

\newenvironment{proof}[1][Proof]{\noindent\textbf{#1.} }{\ \rule{0.5em}{0.5em}}

\newcommand{\bY}{\mathbf{Y}}

\newcommand{\bX}{\mathbf{X}}

\newcommand{\bZ}{\mathbf{Z}}

\newcommand{\RR}{\mathbb{R}}

\newcommand{\Var}{\mathrm{Var}}

\newcommand{\sign}{\mathop{\mathrm{sign}}}

\def\Expt{\mathbb{E}}
\def\b{\mathbf}
\newcommand{\m}{\mathcal}

\def\Expt{\mathbb{E}}
\def\bpi{\mathbbm{\pi}}
\newcommand{\Ber}{\mathrm{Bernoulli}}
\newcommand{\mmse}{\mathsf{MMSE}}
\newcommand{\argmin}{\operatornamewithlimits{argmin}}
\newcommand{\argmax}{\operatornamewithlimits{argmax}}

\begin{document}

\title{Minimum MS. E. Gerber's Lemma}
\author{Or~Ordentlich and
        Ofer~Shayevitz,~\IEEEmembership{Member,~IEEE}
\thanks{The work of O. Ordentlich was supported by the Admas Fellowship Program of the Israel Academy of Science and Humanities. The work of O. Shayevitz was supported by an ERC grant no. 639573, aa CIG grant no. 631983, and an ISF grant no. 1367/14.}}

\parskip 3pt

\maketitle

\begin{abstract}
Mrs. Gerber's Lemma lower bounds the entropy at the output of a binary symmetric channel in terms of the entropy of the input process. In this paper, we lower bound the output entropy via a different measure of input uncertainty, pertaining to the minimum mean squared error (MMSE) prediction cost of the input process. We show that in many cases our bound is tighter than the one obtained from Mrs. Gerber's Lemma. As an application, we evaluate the bound for binary hidden Markov processes, and obtain new estimates for the entropy rate.
\end{abstract}

\section{Introduction}
\label{sec:Intro}

Mrs. Gerber's Lemma~\cite{wz73} lower bounds the entropy of the output of a binary symmetric channel (BSC) in terms of the entropy of the input to the channel. More specifically, if $\bX\in\{0,1\}^n$ is an $n$-dimensional binary random vector with entropy $H(\bX)$, $\bZ\in\{0,1\}^n$ is an $n$-dimensional binary random vector with i.i.d. $\Ber(\alpha)$ components, statistically independent of $\bX$, and $\bY=\bX\oplus\bZ$, Mrs. Gerber's Lemma states that
\begin{align}
\frac{1}{n}H(\bY)\geq h\left(\alpha * h^{-1}\left(\frac{1}{n}H(\bX)\right) \right),\label{eq:MrsGerber}
\end{align}
where $h(p)\triangleq -p\log(p)-(1-p)\log(1-p)$ is the binary entropy function, $h^{-1}(\cdot)$ is its inverse function restricted to $[0,\tfrac{1}{2}]$ and $a* b\triangleq a(1-b)+b(1-a)$ denotes the binary convolution between two numbers $a,b\in[0,1]$. For $\bX$ i.i.d., the inequality~\eqref{eq:MrsGerber} is tight.

The inequality~\eqref{eq:MrsGerber} is in fact a simple consequence of the conditional scalar Mrs. Gerber's Lemma, which states the following: If $U$ is some random variable, $X|U=u\sim\Ber(P_u)$, and $Z\sim\Ber(\alpha)$ is statistically independent of $(X,U)$, we have that
\begin{align}
H(X\oplus Z|U)\geq h\left(\alpha*h^{-1}\left(H(X|U)\right)\right),\label{eq:ScalarMrsGerber}
\end{align}
or alternatively,
\begin{align}
\Expt h(\alpha*P_U)\geq h\left(\alpha*h^{-1}\left(\Expt h(P_U)\right)\right).\label{eq:condMrsGerber}
\end{align}

Since the publication of~\cite{wz73}, many extensions, generalizations and results of a similar flavor have been found, see e.g.,~\cite{Witsenhausen74,cs89,sz90,ja13}. In this paper, we derive a lower bound on entropy of the output $\bY$ in terms of the \emph{minimum mean squared error predictability} of the input $\bX$, as we define next. 

Let $\bpi$ be some permutation of the coordinates $\{1,2,\ldots,n\}$. We define the minimum mean squared error (MMSE) predictability of a binary vector $\bX$ w.r.t. the permutation $\bpi$ as
\begin{align}
&\mmse_{\bpi}(\bX)\nonumber\\ &\triangleq\sum_{i=1}^n \text{MMSE}\left(X_{\pi(i)}|X_{\pi(i-1)},X_{\pi(i-2)},\ldots,X_{\pi(1)}\right)\nonumber\\
&\triangleq\sum_{i=1}^n \Expt\left(\Var\left(X_{\pi(i)}|X_{\pi(i-1)},X_{\pi(i-2)},\ldots,X_{\pi(1)}\right)\right)\nonumber\\
&\triangleq\sum_{i=1}^n \Expt\left(P_i^{\b{\pi}}(1-P_i^{\b{\pi}})\right)\label{eq:defMMSEpi},
\end{align}
where the random variable $P_i^{\b{\pi}}$ is defined as
\begin{align}
P_i^{\b{\pi}}\triangleq \Pr\left(X_{\pi(i)}=1|X_{\pi(i-1)},X_{\pi(i-2)},\ldots,X_{\pi(1)}\right).\label{eq:defPi}
\end{align}
The \emph{worst-case MMSE} predictability of a binary vector $\bX$ is defined as
\begin{align}
\overline{\mmse}(\bX)\triangleq\max_{\bpi}\mmse_{\bpi}(\bX).\label{eq:defMMSE}
\end{align}

Our main result is the following.

\begin{theorem}\label{thm:MMSEGerber}
Let $\bX,\bZ$ be two statistically independent $n$-dimensional random binary vectors, where $\bX$ is arbitrary and $\bZ$ is i.i.d. $\Ber(\alpha)$. Let $\bY=\bX\oplus\bZ$. Then
\begin{align}
\frac{1}{n}H(\bY)\geq h(\alpha)+\left(1-h(\alpha)\right)4\frac{\overline{\mmse}(\bX)}{n},\label{eq:mainBound}
\end{align}
with equality if and only if $\bX$ is memoryless with $\Pr(X_i=1)\in\{0,\tfrac{1}{2},1\}$ for every $i=1,\ldots,n$.
\end{theorem}

In Section~\ref{sec:proof} we prove an MMSE version of the conditional scalar Mrs. Gerber Lemma~\eqref{eq:ScalarMrsGerber}, which implies Theorem~\ref{thm:MMSEGerber} as a simple corollary. In Section~\ref{sec:extensions} we derive several MMSE-based extensions of Theorem~\ref{thm:MMSEGerber}, including a lower bound on $H(\bY)$ for the setting where $\bZ$ is not i.i.d. as well as an upper bound on $H(\bY)$. Section~\ref{sec:comparison} compares our new bound to Mrs. Gerber's Lemma. As an application of Theorem~\ref{thm:MMSEGerber}, in Section~\ref{sec:HMM} we develop a lower bound on the entropy rate of a binary hidden Markov process, which is shown to be considerably stronger than Mrs. Gerber's Lemma in certain scenarios. Furthermore, our MMSE-based scalar lower bound derived is combined with a bounding technique developed in~\cite{ow11} to obtain new estimates on the entropy rate of binary hidden Markov processes.

\section{Proofs}
\label{sec:proof}

Mrs. Gerber's Lemma is proved by first deriving the conditional scalar inequality~\eqref{eq:ScalarMrsGerber} and then invoking the chain rule for entropy and convexity of the function $g(u)=h(\alpha*h^{-1}(u))$ to arrive at~\eqref{eq:MrsGerber}, see~\cite{wz73,elgamalkim}. Similarly, we begin by proving an MMSE version of~\eqref{eq:ScalarMrsGerber} below, from which Theorem~\ref{thm:MMSEGerber} will follow as a simple corollary.

\begin{lemma}\label{lem:condMMSEGerber}
Let $U$ be a random variable and let $X|U=u\sim\Ber(P_u)$. Denote the MMSE in estimating $X$ from $U$ by
\begin{align}
\mathrm{MMSE}(X|U)\triangleq \Expt\left(\Var(X|U)\right)=\Expt\left(P_U(1-P_U)\right).\label{eq:MMSExu}
\end{align}
Let $Z\sim\Ber(\alpha)$ be statistically independent of $(X,U)$.
Then
\begin{align}
H(X\oplus Z|U)\geq h(\alpha)+\left(1-h(\alpha)\right)4\mathrm{MMSE}(X|U),\nonumber
\end{align}
with equality if and only if $P_u\in\{0,\tfrac{1}{2},1\}$ for any value of $u$.
\end{lemma}

\begin{proof}
Since $Z$ is statistically independent of $(X,U)$ we have
\begin{align}
H(X\oplus Z|U)=\Expt h(P_U*\alpha).\label{eq:condEnt}
\end{align}
Let $V_U\triangleq P_U-\tfrac{1}{2}$ and note that
\begin{align}
P_U*\alpha&=\left(\frac{1}{2}+V_U\right)(1-\alpha)+\alpha\left(\frac{1}{2}-V_U\right)\nonumber\\
&=\frac{1}{2}+(1-2\alpha)V_U.\label{eq:conviden}
\end{align}
Recall that the Taylor series expansion of the binary entropy function around $\tfrac{1}{2}$ is
\begin{align}
h\left(\frac{1}{2}+\frac{p}{2}\right)=1-\sum_{k=1}^{\infty}\frac{\log(e)}{2k(2k-1)}p^{2k},\label{eq:entTaylor}
\end{align}
and therefore, by~\eqref{eq:conviden} we have
\begin{align}
h(P_U&*\alpha)=1-\sum_{k=1}^{\infty}\frac{\log(e)}{2k(2k-1)}(1-2\alpha)^{2k}(2V_U)^{2k}\nonumber\\
&\geq 1-4 V_U^2\sum_{k=1}^{\infty}\frac{\log(e)}{2k(2k-1)}(1-2\alpha)^{2k}\label{eq:VuIneq}\\
&=1-4 V_U^2+4 V_U^2\left(1-\sum_{k=1}^{\infty}\frac{\log(e)}{2k(2k-1)}(1-2\alpha)^{2k}\right)\nonumber\\
&=1-4 V_U^2+4 V_U^2\cdot h\left(\frac{1}{2}+\frac{1-2\alpha}{2} \right)\nonumber\\
&=1-4 V_U^2\left(1-h(\alpha)\right)\label{eq:convent},
\end{align}
where~\eqref{eq:VuIneq} follows from the fact that $|2V_U|\leq 1$, and is satisfied with equality if and only if $V_U\in\{-\tfrac{1}{2},0,\tfrac{1}{2}\}$, which implies that $P_U\in\{0,\tfrac{1}{2},1\}$. Substituting~\eqref{eq:convent} into~\eqref{eq:condEnt} gives
\begin{align}
H(X\oplus Z|U)&\geq 1-\left(1-h(\alpha)\right) \ 4\Expt(V^2_U)\nonumber\\
&=1-\left(1-h(\alpha)\right) \ 4\Expt\left(\frac{1}{2}-P_U\right)^2\nonumber\\
&=h(\alpha)+\left(1-h(\alpha)\right) \ 4\Expt\left(P_U(1-P_U)\right),\nonumber
\end{align}
as desired.
\end{proof}

\begin{remark}
Note that the only property of the binary entropy function used in the proof above is that all coefficients of (nonzero) even order in its Taylor expansion around $\tfrac{1}{2}$ are negative, whereas all odd coefficients are zero. It follows that for any function $g:[0,1]\mapsto\RR$ whose Taylor expansion around $\tfrac{1}{2}$ is of the form
\begin{align}
g\left(\frac{1}{2}+\frac{p}{2}\right)=c_0-\sum_{k=1}^{\infty}c_k (p)^{2k},\nonumber
\end{align}
where $c_k\geq 0$ for all positive $k$ we have
\begin{align}
\Expt g\left(\alpha*P_U\right)\geq g(\alpha)+\left(c_0-g(\alpha)\right)4\mathrm{MMSE}(X|U).\nonumber
\end{align}
\end{remark}

Theorem~\ref{thm:MMSEGerber} now follows as a straightforward corollary of Lemma~\ref{lem:condMMSEGerber}.

\begin{proof}[Proof of Theorem~\ref{thm:MMSEGerber}]
By the chain rule for entropy, for any permutation $\bpi$ we have
\begin{align}
&H(\bY)=\sum_{i=1}^n H\left(Y_{\pi(i)}|Y_{\pi(i-1)},\ldots,Y_{\pi(1)}\right)\nonumber\\
&=\sum_{i=1}^n H\left(X_{\pi(i)}\oplus Z_{\pi(i)}|Y_{\pi(i-1)},\ldots,Y_{\pi(1)}\right)\label{eq:chain}\\
&\geq \sum_{i=1}^n h(\alpha)+\left(1-h(\alpha)\right) 4\mathrm{MMSE}\left(X_{\pi(i)}|Y_{\pi(i-1)},\ldots,Y_{\pi(1)}\right)\label{eq:Ymmse}.
\end{align}
Clearly
\begin{align}
&\mathrm{MMSE}\left(X_{\pi(i)}|Y_{\pi(i-1)},\ldots,Y_{\pi(1)}\right)\nonumber\\
&\geq \mathrm{MMSE}\left(X_{\pi(i)}|Y_{\pi(i-1)},\ldots,Y_{\pi(1)},Z_{\pi(i-1)},\ldots,Z_{\pi(1)}\right)\nonumber\\
&=\mathrm{MMSE}\left(X_{\pi(i)}|X_{\pi(i-1)},\ldots,X_{\pi(1)},Z_{\pi(i-1)},\ldots,Z_{\pi(1)}\right)\nonumber\\
&=\mathrm{MMSE}\left(X_{\pi(i)}|X_{\pi(i-1)},\ldots,X_{\pi(1)}\right),\label{eq:MMSEineq}
\end{align}
where the last equality follows since the random variables $\{Z_i\}_{i=1}^n$ are statistically independent of $\{X_i\}_{i=1}^n$. Thus, for any permutation $\bpi$ we have
\begin{align}
H(\bY)&\geq n h(\alpha)\nonumber\\
&+\left(1-h(\alpha)\right) \ 4 \sum_{i=1}^n \mathrm{MMSE}\left(X_{\pi(i)}|X_{\pi(i-1)},\ldots,X_{\pi(1)}\right),\label{eq:MMSEnonOpt}
\end{align}
and~\eqref{eq:mainBound} follows by maximizing~\eqref{eq:MMSEnonOpt} w.r.t. $\bpi$. By Lemma~\ref{lem:condMMSEGerber}, the inequality~\eqref{eq:Ymmse} is tight if and only if $\Pr\left(X_{\pi(i)}=1|Y_{\pi(i-1)},\ldots,Y_{\pi(1)}\right)\in\{0,\tfrac{1}{2},1\}$ for every $i$ and every realization of the vector $\left(Y_{\pi(i-1)},\ldots,Y_{\pi(1)}\right)$, whereas for $0<\alpha<1$ the inequality~\eqref{eq:MMSEineq} is tight if and only if $\bX$ is memoryless. Thus,~\eqref{eq:mainBound} holds with equality if and only if $\bX$ is memoryless with $\Pr(X_i=1)\in\{0,\tfrac{1}{2},1\}$ for every $i=1,\ldots,n$.
\end{proof}

\section{Extensions}
\label{sec:extensions}

In this section we derive several simple extensions of our main results. Since the proofs are quite similar to those of Lemma~\ref{lem:condMMSEGerber} and Theorem~\ref{thm:MMSEGerber}, we omit the full details and only sketch the differences instead.

We begin with a straightforward extension of Theorem~\ref{thm:MMSEGerber} to the conditional entropy $H(\bY|W)$ where $\bX$ may depend on $W$, while $\bZ$ and $W$ are statistically independent.

\begin{theorem}\label{thm:ConditionalEntropy}
Let $W$ be some random variable, and let $\bX,\bZ$ be two $n$-dimensional random binary vectors, where $\bX$ is arbitrary and $\bZ$ is i.i.d. $\Ber(\alpha)$. Assume that $(\bX,W)$ is mutually independent of $\bZ$, and let $\bY=\bX\oplus\bZ$. Then
\begin{align}
\frac{1}{n}H(\bY|W)\geq h(\alpha)+\left(1-h(\alpha)\right)4\frac{\overline{\mmse}(\bX|W)}{n},\nonumber
\end{align}
with equality if and only if $\bX|W=w$ is memoryless with $\Pr(X_i=1|W=w)\in\{0,\tfrac{1}{2},1\}$ for every $i=1,\ldots,n$ and every $w$.
\end{theorem}

\begin{proof}
The proof is omitted as it follows the exact same steps as in the proof of Theorem~\ref{thm:MMSEGerber}, where the conditioning on $W$ is added where relevant.
\end{proof}

Next, we show that our lower bound can also be extended to the case of a binary noisy channel with memory. To that end, we first need to derive a simple generalization of Lemma~\ref{lem:condMMSEGerber}.

\begin{lemma}\label{lem:MMSEstationaryNoise}
Let $U=(T,W)$, where $T$ and $W$ are statistically independent. Let $X$ and $Z$ be conditionally independent given $U$, such that $X|U=(t,w)\sim\Ber(P_{t})$ and $Z|U=(t,w)\sim\Ber(\alpha_{w})$. Let $\mathrm{MMSE}(X|U)=\mathrm{MMSE}(X|T)$ be as defined in~\eqref{eq:MMSExu}. Then
\begin{align}
H(X\oplus Z|U)\geq H(Z|W)+\left(1-H(Z|W)\right)4\mathrm{MMSE}(X|T),\nonumber
\end{align}
with equality if and only if $P_{t}\in\{0,\tfrac{1}{2},1\}$ for any value of $t$.
\end{lemma}

\begin{proof}[Sketch of proof]
The proof follows the same lines as the proof of Lemma~\ref{lem:condMMSEGerber}. Since $T$ and $W$ are statistically independent, we have $H(X\oplus Z|U)=\Expt h(P_T*\alpha_W)$. By~\eqref{eq:convent} we have that
\begin{align}
h(P_T*\alpha_W)\geq 1-4 \left(\frac{1}{2}-P_T\right)^2\left(1-h(\alpha_W)\right)\nonumber,
\end{align}
We therefore have
\begin{align}
\Expt_U h(P_T*\alpha_W)&\geq \Expt_U \left(1-4 \left(\frac{1}{2}-P_T\right)^2\left(1-h(\alpha_W)\right)\right)\nonumber\\
&=1-4\Expt_T \left(\frac{1}{2}-P_T\right)^2\left(1-\Expt_W h(\alpha_W)\right),\nonumber
\end{align}
and the lemma follows by recalling that $4\Expt_T \left(\frac{1}{2}-P_T\right)^2=1-4\mathrm{MMSE}(X|T)$ and that $\Expt_W h(\alpha_W)=H(Z|W)$.
\end{proof}

As a simple corollary, we obtain the following.
\begin{theorem}\label{thm:StationaryNoise}
Let $\bX,\bZ$ be two statistically independent $n$-dimensional random binary vectors, and let $\bY=\bX\oplus\bZ$. Then
\begin{align}
H(\bY)&\geq \max_{\bpi}\bigg\{H(\bZ)+4{\mmse}_{\bpi}(\bX)\nonumber\\
&-4\sum_{i=1}^n H\left(Z_{\pi(i)}|Z_{\pi(i-1)},\ldots,Z_{\pi(1)}\right)\nonumber\\
& \ \ \ \ \ \ \ \   \cdot\mathrm{MMSE}\left(X_{\pi(i)}|X_{\pi(i-1)},\ldots,X_{\pi(1)}\right)\bigg\},\nonumber
\end{align}
with equality if and only if $\bZ$ is memoryless and $\bX$ is memoryless with $\Pr(X_i=1)\in\{0,\tfrac{1}{2},1\}$ for every $i=1,\ldots,n$.
\end{theorem}

\begin{proof}
By the chain rule for entropy, for any permutation $\bpi$ we have
\begin{align}
&H(\bY)=\sum_{i=1}^n H\left(Y_{\pi(i)}|Y_{\pi(i-1)},\ldots,Y_{\pi(1)}\right)\nonumber\\
&\geq\sum_{i=1}^n H\left(Y_{\pi(i)}|X_{\pi(i-1)},\ldots,X_{\pi(1)},Z_{\pi(i-1)},\ldots,Z_{\pi(1)}\right)\label{eq:ineqTQ}\\
&=\sum_{i=1}^n H\left(X_{\pi(i)}\oplus Z_{\pi(i)}|T_{\bpi}^i,W_{\bpi}^i\right)\label{eq:chain2}
\end{align}
where the random variables
\begin{align}
T_{\bpi}^i&\triangleq\left(X_{\pi(i-1)},\ldots,X_{\pi(1)}\right)\nonumber\\
W_{\bpi}^i&\triangleq\left(Z_{\pi(i-1)},\ldots,Z_{\pi(1)}\right)\nonumber
\end{align}
are statistically independent, and $X_{\pi(i)}$ and $Z_{\pi(i)}$ are conditionally independent given $(T_{\bpi}^i,W_{\bpi}^i)$,  since $\bX$ and $\bZ$ are statistically independent. The inequality~\eqref{eq:ineqTQ} is tight if and only if $\bX$ and $\bY$ are both memoryless. Now, by Lemma~\ref{lem:MMSEstationaryNoise} we have that
\begin{align}
&H\bigg(X_{\pi(i)}\oplus Z_{\pi(i)}|T_{\bpi}^i,W_{\bpi}^i\bigg)\nonumber\\
&\geq H(Z_{\pi(i)}|W^i_{\bpi})
+\left(1-H(Z_{\pi(i)}|W^i_{\bpi})\right)4\mmse\left(X_{\pi(i)}|T_{\bpi}^i\right).\nonumber
\end{align}
Summing over $i$ gives the desired result.
\end{proof}

A simple consequence of Theorem~\ref{thm:StationaryNoise} is that if $\bX$ and $\bZ$ are statistically independent binary symmetric first-order Markov processes with transition probabilities $q_1$ and $q_2$, respectively, then $\tfrac{1}{n}H(\bY)\geq h(q_1)+4q_2(1-q_2)(1-h(q_1))$. This bound uses the identity permutation $\bpi=(1,\ldots,n)$. We note that a more clever choice of $\bpi$, as used in Section~\ref{sec:HMM}, can result in a better bound.

We end this section by deriving an upper bound on $H(\bY)$ in terms of the best-case MMSE predictability of $\bX$ from $\bY$
\begin{align}
\underline{\mmse}(\bX|\bY)\triangleq \min_{\bpi}\sum_{i=1}^n\mathrm{MMSE}\left(X_{\pi(i)}|Y_{\pi(i-1)},\ldots,Y_{\pi(1)}\right).\nonumber
\end{align}
To that end, we first upper bound $H(X\oplus Z|U)$ in terms of $\text{MMSE}(X|U)$.
\begin{lemma}\label{lem:upperMMSE}
Let $U$ be some random variable and let $X|U=u\sim\Ber(P_u)$. Let $Z\sim\Ber(\alpha)$ be statistically independent of $(X,U)$.
Then
\begin{align}
H(X\oplus Z|U)\leq h\left(\frac{1}{2}+\frac{1-2\alpha}{2}\sqrt{1-4\mathrm{MMSE}(X|U)} \right),\nonumber
\end{align}
with equality if and only if $\left|P_u-\frac{1}{2}\right|$ does not depend on $u$.
\end{lemma}

\begin{proof}
Define the function $Q(t)\triangleq h\left(\tfrac{1}{2}+\sqrt{t}\right)$ and note that it is concave over $[0,\tfrac{1}{4}]$. By~\eqref{eq:condEnt} and~\eqref{eq:conviden} we have
\begin{align}
H(X\oplus Z|U)&=\Expt h\left(\frac{1}{2}+(1-2\alpha)\left(P_U-\frac{1}{2}\right) \right)\nonumber\\
&=\Expt h\left(\frac{1}{2}+\sqrt{(1-2\alpha)^2\left(P_U-\frac{1}{2}\right)^2} \right)\nonumber\\
&\leq h\left(\frac{1}{2}+\sqrt{\Expt\left[(1-2\alpha)^2\left(P_U-\frac{1}{2}\right)^2\right]} \right)\nonumber\\
&= h\left(\frac{1}{2}+(1-2\alpha)\sqrt{\Expt\left(\frac{1}{4}-P_u(1-P_U)\right)} \right)\nonumber\\
&= h\left(\frac{1}{2}+\frac{1-2\alpha}{2}\sqrt{1-4\text{MMSE}(X|U)} \right)\nonumber,
\end{align}
as desired.
\end{proof}

\begin{remark}
In the special case where $\alpha=0$, Lemma~\ref{lem:upperMMSE} reduces to the inequality
\begin{align}
\Expt h(P_U)\leq h\left(\frac{1}{2}+\sqrt{\Expt\left(\frac{1}{2}-P_U)^2 \right)} \right),\nonumber
\end{align}
which was obtained by Wyner in~\cite[eq. (3.11)]{Wyner75}
\label{rem:entUpper}
\end{remark}

The function $F_{\alpha}(x)\triangleq h\left(\tfrac{1}{2}+\tfrac{1-2\alpha}{2}\sqrt{1-4x}\right)$ is concave and monotone non-decreasing for $x\in[0,\tfrac{1}{4}]$ and any value of $\alpha\in[0,\tfrac{1}{2}]$. Combining this with~\eqref{eq:chain} and with Lemma~\ref{lem:upperMMSE} gives the following.
\begin{theorem}\label{thm:uperVec}
Let $\bX,\bZ$ be two statistically independent $n$-dimensional random binary vectors, where $\bX$ is arbitrary and $\bZ$ is i.i.d. $\Ber(\alpha)$. Let $\bY=\bX\oplus\bZ$. Then
\begin{align}
\frac{1}{n}H(\bY)\leq h\left(\frac{1}{2}+\frac{1-2\alpha}{2}\sqrt{1-4\frac{\underline{\mmse}(\bX|\bY)}{n}}\right),\nonumber
\end{align}
with equality if and only if $\bX$ is i.i.d.
\end{theorem}

\section{Comparison with Mrs. Gerber's Lemma}
\label{sec:comparison}

In this section we compare the performance of our MMSE-based bound to Mrs. Gerber's Lemma. First, we consider the family of random vectors with fixed $\overline{\mmse}(\bX)$. Clearly, the bound from Theorem~\ref{thm:MMSEGerber} is the same for all members of this family. However, the entropy $H(\bX)$ may vary within the family, and hence applying Mrs. Gerber's Lemma results in a range of bounds, which can be juxtaposed with the bound of Theorem~\ref{thm:MMSEGerber}. Similarly, we fix $H(\bX)$ and juxtapose Mrs. Gerber's Lemma with the range of bounds obtained by applying Theroem~\ref{thm:MMSEGerber}.

For the special case of $\alpha=0$, Theorem~\ref{thm:MMSEGerber} reads
\begin{align}
H(\bX)\geq 4\overline{\mmse}(\bX).\label{eq:EntMMSELower}
\end{align}
and Theorem~\ref{thm:uperVec} reads
\begin{align}
H(\bX)\leq n h\left(\frac{1}{2}+\frac{1}{2}\sqrt{1-4\frac{\underline{\mmse}(\bX)}{n}} \right)\nonumber\\
\leq n h\left(\frac{1}{2}+\frac{1}{2}\sqrt{1-4\frac{\overline{\mmse}(\bX)}{n}} \right).\label{eq:EntMMSEUpper}
\end{align}

Denote the RHS of~\eqref{eq:condMrsGerber} by
\begin{align}
\mathsf{MGL}\left(\alpha,P_{\bX}\right)\triangleq h\left(\alpha*h^{-1}\left(\frac{H(\bX)}{n}\right) \right).\nonumber
\end{align}
and the RHS of~\eqref{eq:mainBound} by
\begin{align}
\mathsf{NEW}\left(\alpha,P_{\bX}\right)\triangleq h(\alpha)+\left(1-h(\alpha)\right)4\frac{\overline{\mmse}(\bX)}{n},\nonumber
\end{align}
By~\eqref{eq:EntMMSEUpper} and~\eqref{eq:EntMMSELower} it follows that
\begin{align}
h\bigg(\alpha*&h^{-1}\left(4\frac{\overline{\mmse}(\bX)}{n}\right) \bigg)\leq \mathsf{MGL}\left(\alpha,P_{\bX}\right)   \nonumber\\
&\leq  h\left(\alpha*\left(\frac{1}{2}+\frac{1}{2}\sqrt{1-4\frac{\overline{\mmse}(\bX)}{n}} \right) \right)\label{eq:MGLbounds}
\end{align}
Figure~\ref{fig:MGLcomp} depicts the lower and upper bound on $\mathsf{MGL}\left(\alpha,P_{\bX}\right)$ from~\eqref{eq:MGLbounds} as a function of $\overline{\mmse}(\bX)$ along with $\mathsf{NEW}\left(\alpha,P_{\bX}\right)$, for $\alpha=0.11$. It is seen that for all values of $\overline{\mmse}(\bX)$ our bound is quite close to the upper bound on $\mathsf{MGL}\left(\alpha,P_{\bX}\right)$, and is often significantly stronger than the lower bound on $\mathsf{MGL}\left(\alpha,P_{\bX}\right)$. In general, for small values of $\alpha$, $\mathsf{NEW}\left(\alpha,P_{\bX}\right)$ will be close to the lower bound on $\mathsf{MGL}\left(\alpha,P_{\bX}\right)$ and will approach the upper bound on $\mathsf{MGL}\left(\alpha,P_{\bX}\right)$ as $\alpha$ increases. Figure~\ref{fig:MGLcompAlpha} demonstrates this phenomenon for $4\overline{\mmse}(\bX)=0.5$.

\begin{figure*}[!ht]
\begin{center}
\subfloat[$\alpha=0.11$]{
\includegraphics[width=0.45\textwidth]{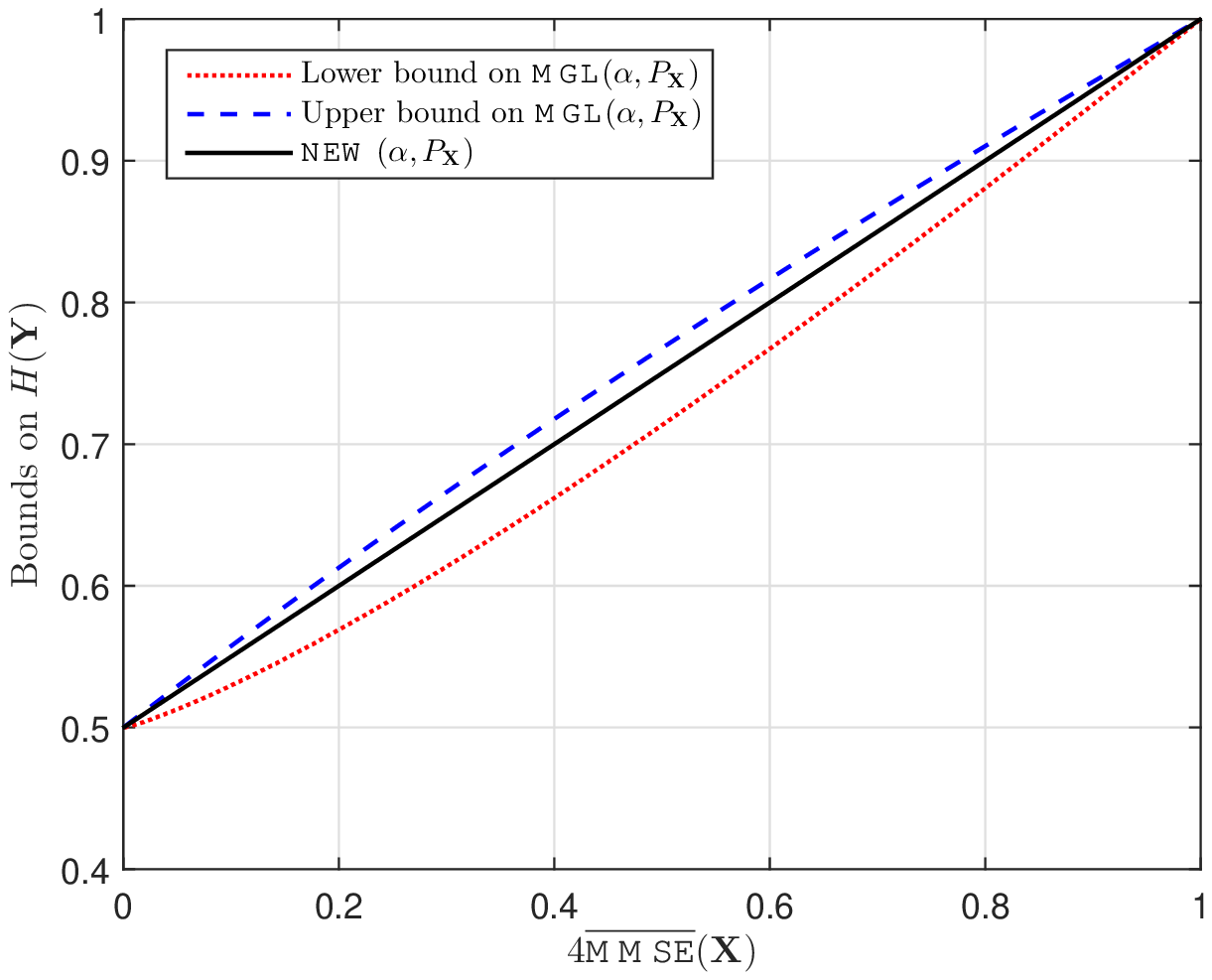}
\label{fig:MGLcomp}}
\qquad
\subfloat[$\overline{\mmse}(\bX)=0.5$]{
\includegraphics[width=0.45\textwidth]{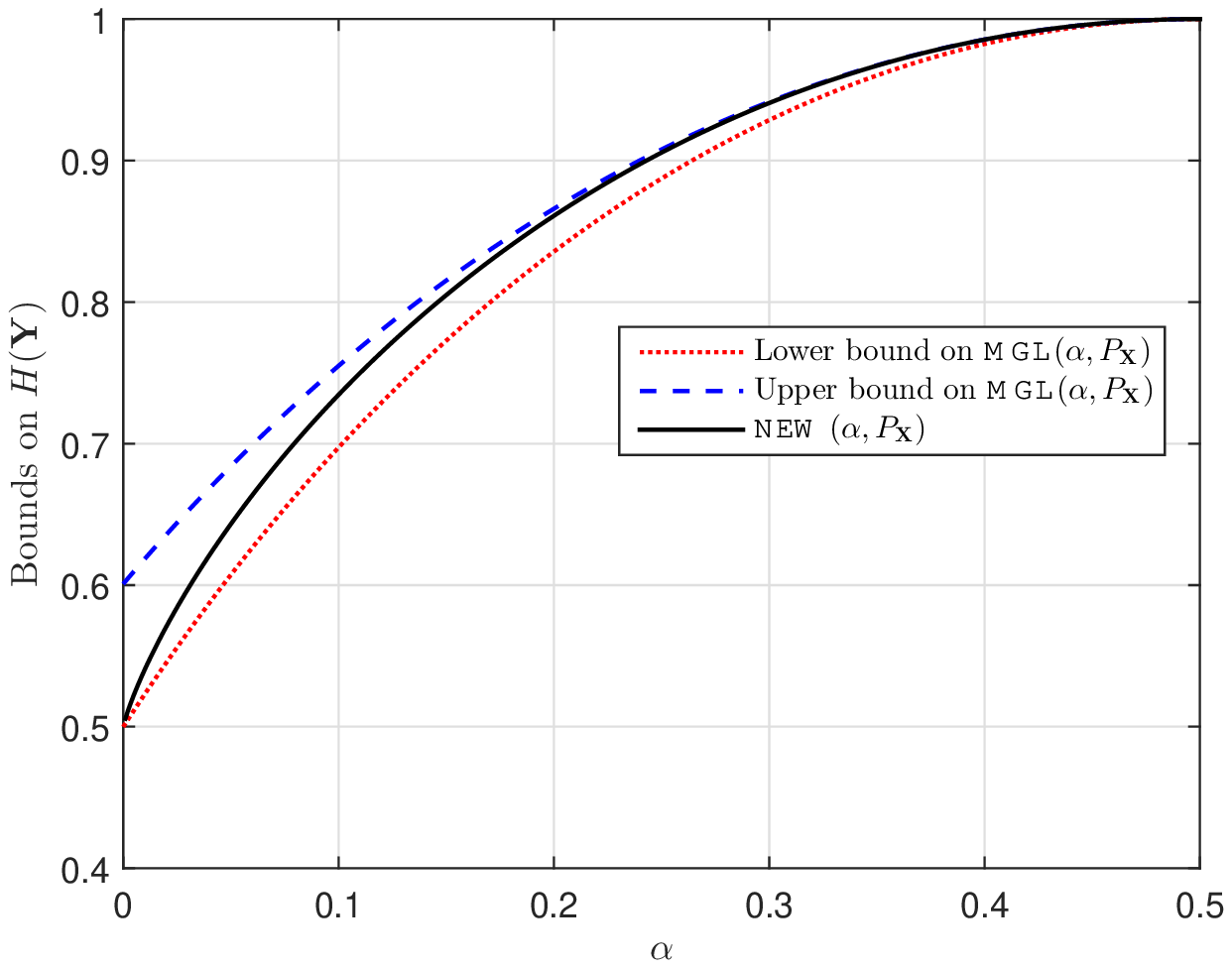}
\label{fig:MGLcompAlpha}}

\end{center}
\caption{Comparison between the lower and upper bounds on $\mathsf{MGL}\left(\alpha,P_{\mathbf{X}}\right)$ from~\eqref{eq:MGLbounds} and $\mathsf{NEW}\left(\alpha,P_{\bX}\right)$.}
\end{figure*}

Equivalently, by~\eqref{eq:EntMMSELower} and~\eqref{eq:EntMMSEUpper}, we also have that
\begin{align}
4n h^{-1}&\left(\frac{H(\bX)}{n}\right)\left(1-h^{-1}\left(\frac{H(\bX)}{n}\right)\right)\nonumber\\
&\leq 4\overline{\mmse}(\bX)
\leq H(\bX).\label{eq:MMSEbounds}
\end{align}
In fact,~\eqref{eq:MMSEbounds} holds for $4\mmse_{\bpi}(\bX)$ with any permutation $\bpi$, and implies
\begin{align}
h(\alpha)+&(1-h(\alpha)4 h^{-1}\left(\frac{H(\bX)}{n}\right)\left(1-h^{-1}\left(\frac{H(\bX)}{n}\right)\right)\nonumber\\
&\leq \mathsf{NEW}\left(\alpha,P_{\bX}\right)
\leq h(\alpha)+(1-h(\alpha)H(\bX)\label{eq:NEWbounds}
\end{align}

Figure~\ref{fig:Newcomp} depicts the lower and upper bound on $\mathsf{NEW}\left(\alpha,P_{\bX}\right)$ from~\eqref{eq:NEWbounds} as a function of $H(\bX)$ along with $\mathsf{MGL}\left(\alpha,P_{\mathbf{X}}\right)$, for $\alpha=0.11$. It is seen that for all values of $H(\bX)$, $\mathsf{MGL}\left(\alpha,P_{\mathbf{X}}\right)$ is quite close to the lower bound on $\mathsf{NEW}\left(\alpha,P_{\bX}\right)$, and is often significantly weaker than the upper bound on $\mathsf{NEW}\left(\alpha,P_{\bX}\right)$. In general, for small values of $\alpha$, $\mathsf{MGL}\left(\alpha,P_{\mathbf{X}}\right)$ will be close to the upper bound on $\mathsf{NEW}\left(\alpha,P_{\bX}\right)$ and will approach the lower bound on $\mathsf{NEW}\left(\alpha,P_{\bX}\right)$ as $\alpha$ increases. Figure~\ref{fig:NewcompAlpha} demonstrates this phenomenon for $H(\bX)=0.5$.

\begin{figure*}[!ht]
\begin{center}
\subfloat[$\alpha=0.11$]{
\includegraphics[width=0.45\textwidth]{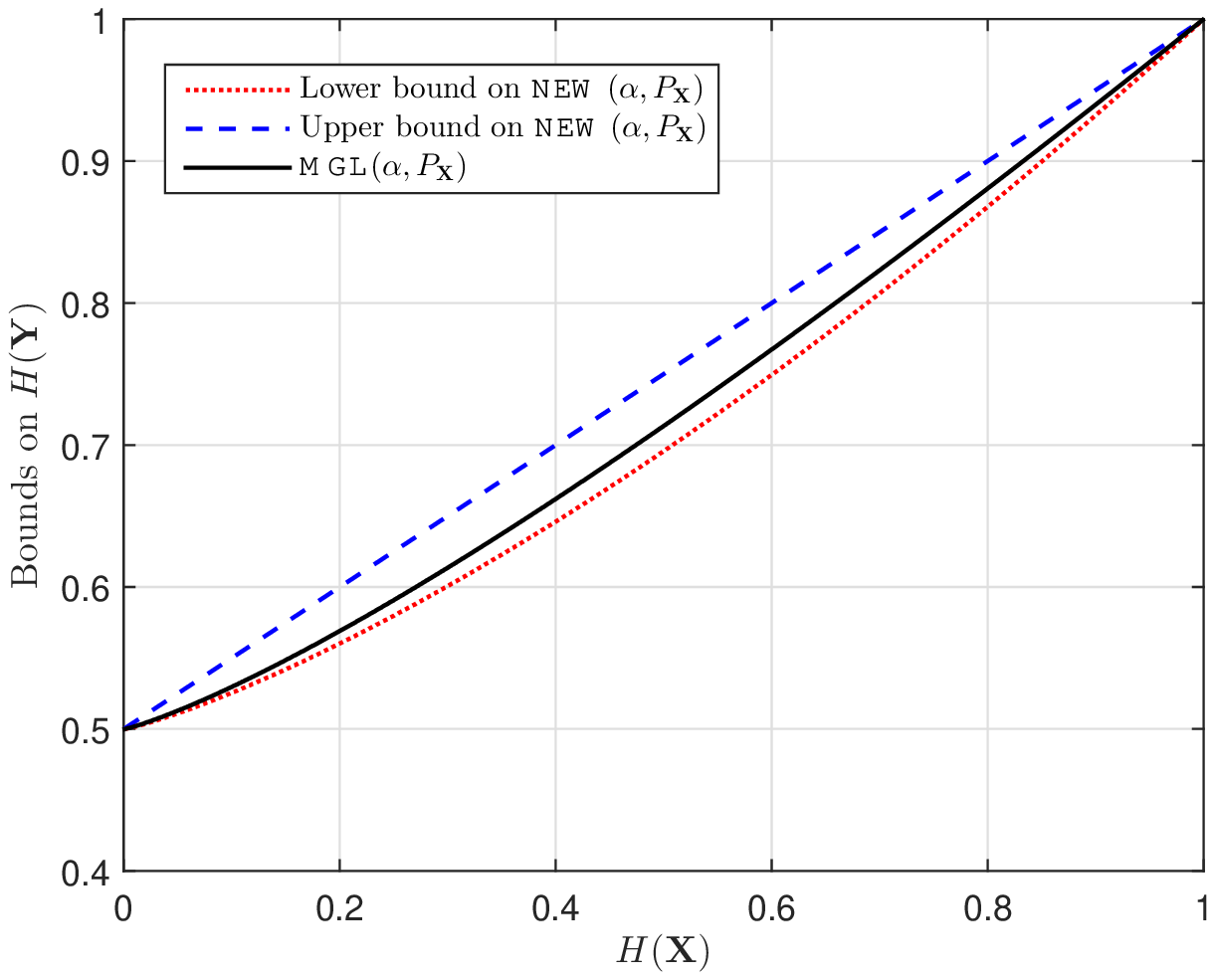}
\label{fig:Newcomp}}
\qquad
\subfloat[$H(\bX)=0.5$]{
\includegraphics[width=0.45\textwidth]{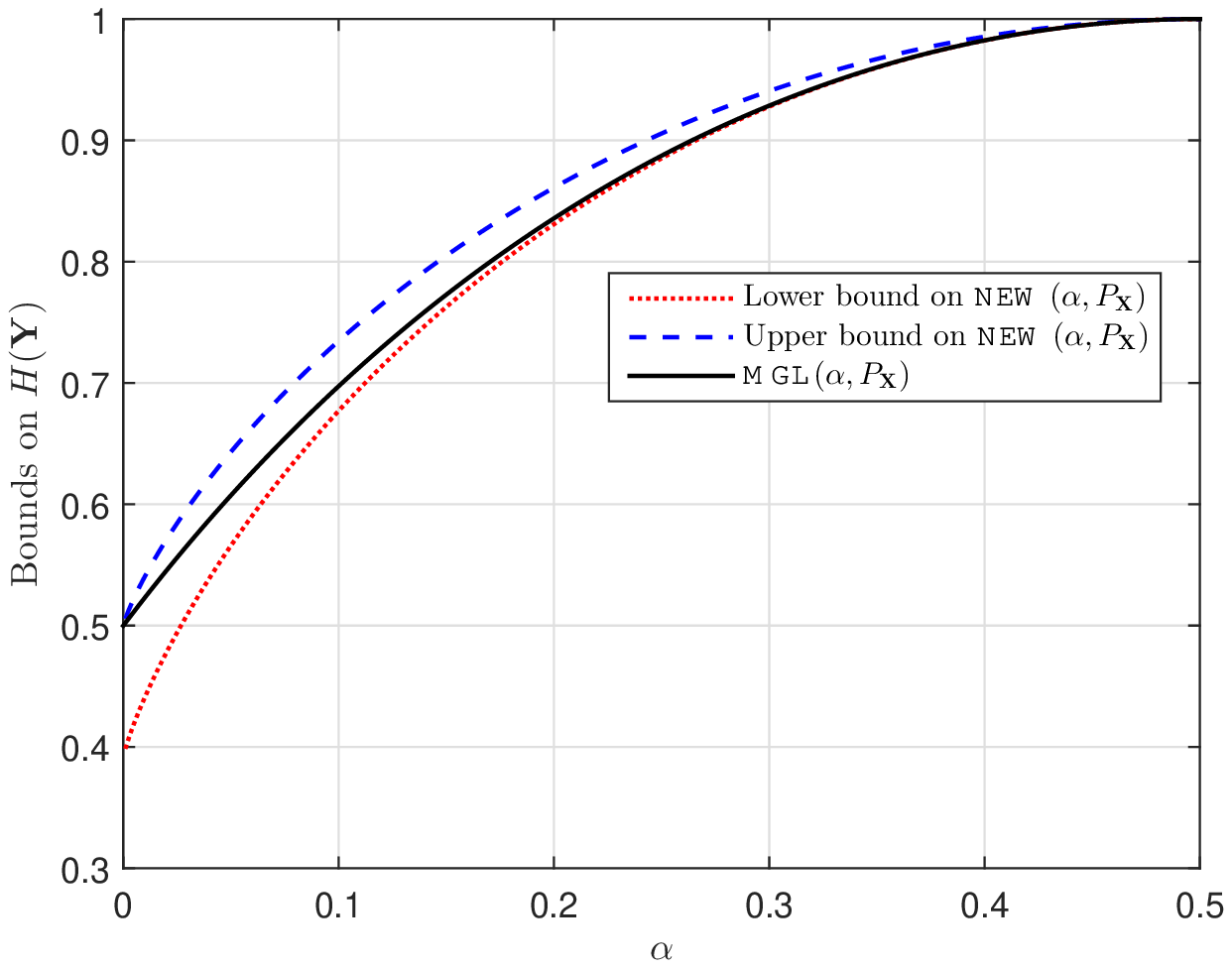}
\label{fig:NewcompAlpha}}

\end{center}
\caption{Comparison between the lower and upper bounds on $\mathsf{NEW}\left(\alpha,P_{\mathbf{X}}\right)$ from~\eqref{eq:NEWbounds} and $\mathsf{MGL}\left(\alpha,P_{\mathbf{X}}\right)$.}
\end{figure*}

\section{Application: Lower Bound on the Entropy Rate of a Binary Hidden Markov Process}
\label{sec:HMM}

In this section we apply Theorem~\ref{thm:MMSEGerber} to derive a simple lower bound on the entropy rate of a binary hidden Markov process. Let $X_1\sim\Ber\left(\frac{1}{2}\right)$ and for $m=2,3,\ldots$ let $X_m=X_{m-1}\oplus W_m$ where $\{W_m\}$ is an i.i.d. $\Ber(q)$ process statistically independent of $X_1$. Clearly, the process $\{X_n\}$ is a symmetric first-order Markov Process. We define the hidden Markov process $Y_n=X_n\oplus Z_n$, where $\{Z_n\}$ is an i.i.d. $\Ber(\alpha)$ process statistically independent of the process $\{X_n\}$. Our goal in this section is to derive a lower bound on the entropy rate of $\{Y_n\}$ defined as
\begin{align}
\overline{H}(Y)\triangleq\lim_{n\rightarrow\infty}\frac{H(Y_1,\ldots,Y_n)}{n}.
\end{align}
One very simple bound can be obtained by noting that $\overline{H}(X)=h(q)$ and applying Mrs. Gerber's Lemma~\eqref{eq:MrsGerber} which gives
\begin{align}
\overline{H}(Y)\geq h(\alpha*q).\label{eq:HMPgerber}
\end{align}
We will see that in many cases our MMSE-based bound from Theorem~\ref{thm:MMSEGerber} provides tighter bounds.

Note that for any $\bpi$ it holds that $\overline{\mmse}(\bX)\geq\mmse_{\bpi}(\bX)$ and therefore Theorem~\ref{thm:MMSEGerber} implies that for any choice of $\bpi$
\begin{align}
\frac{1}{n}H(\bY)\geq h(\alpha)+\left(1-h(\alpha)\right)4\frac{\mmse_{\bpi}(\bX)}{n}.\label{eq:MMSEpi}
\end{align}
Thus, in order to apply Theorem~\ref{thm:MMSEGerber} we need to choose some $\bpi$ and evaluate $\mmse_{\bpi}(\bX)$. A trivial choice is the identity $\bpi=\{1,2,\ldots,n\}$, for which $\tfrac{\mmse_{\bpi}(\bX)}{n}=q(1-q)$ and our bound yields $\overline{H}(Y)\geq h(\alpha)+(1-h(\alpha))4q(1-q)$. It is easy to see that this choice of $\bpi$ yields the lower bound on $\mathsf{NEW}\left(\alpha,P_{\mathbf{X}}\right)$ from~\eqref{eq:NEWbounds}, and is therefore strictly weaker than~\eqref{eq:HMPgerber}. We would therefore like to choose a permutation $\bpi$ that will incur a higher value of $\mmse_{\bpi}(\bX)$. Assume that $\log{n}$ is an integer. A natural candidate is the following
\begin{align}
\bpi=\left(n,\frac{n}{2},\frac{n}{4},\frac{3n}{4},\frac{n}{8},\frac{3n}{8},\frac{5n}{8},\frac{7n}{8},\frac{n}{16},\frac{3n}{16},\ldots \right).\label{eq:piMarkov}
\end{align}
With this choice of $\pi$ we have that if $\pi(i)=r n/2^k$, for $r=1,3,\ldots,2^k-1$, we have that
\begin{align}
\text{MMSE}&\left(X_{\pi(i)}|X_{\pi(i-1)},X_{\pi(i-2)},\ldots,X_{\pi(1)}\right)\nonumber\\
&\geq\text{MMSE}\left(X_{\frac{r n}{2^k}}|X_{\frac{r n}{2^k}-\frac{n}{2^k}},X_{\frac{r n}{2^k}+\frac{n}{2^k}}\right)\nonumber\\
&=\text{MMSE}\left(X_m|X_{m-\frac{n}{2^k}},X_{m+\frac{n}{2^k}}\right)\nonumber\\
&\triangleq\text{MMSE}\left(\frac{n}{2^k}\right)\nonumber,
\end{align}
where the inequality follows from the Markovity of $\{X_n\}$ which implies that the conditional distribution of $X_m$ given multiple samples from the past and the future of the process depends only on the nearest sample from the past and the nearest sample from the future. We therefore have
\begin{align}
\frac{\mmse_{\bpi}(\bX)}{n}&\geq \sum_{k=1}^{\log{n}}\frac{1}{2}\frac{2^k}{n}\text{MMSE}\left(\frac{n}{2^k}\right)\nonumber\\
&=\frac{1}{2}\sum_{k=1}^{\log{n}}2^{-(\log{n}-k)}\text{MMSE}\left(2^{\log{n}-k}\right)\nonumber\\
&=\frac{1}{2}\sum_{t=0}^{\log{n}-1}2^{-t}\text{MMSE}\left(2^{t}\right).\label{eq:sumMMSE}
\end{align}
It now only remains to calculate
\begin{align}
\text{MMSE}\left(\ell\right)&=\text{MMSE}(X_n|X_{n-\ell}X_{n+\ell})\nonumber\\
&=\Expt\left(P_1^{\ell}(X_{n+\ell},X_{n-\ell})P_0^{\ell}(X_{n+\ell},X_{n-\ell}))\right)
\end{align}
where the random variable $P_i^\ell(X_{n+\ell},X_{n-\ell})$ is defined as
\begin{align}
P_i^\ell&(x_{n+\ell},x_{n-\ell})\triangleq\Pr(X_n=i|X_{n-\ell}=x_{n-\ell},X_{n+\ell}=x_{n+\ell})\nonumber\\
&=\frac{P(X_{n+\ell}=x_{n+\ell},X_n=i,X_{n-\ell}=x_{n-\ell})}{P(X_{n-\ell}=x_{n-\ell},X_{n+\ell}=x_{n+\ell})}\nonumber\\
&=\frac{P(X_{n+\ell}=x_{n+\ell}|X_n=i)P(X_n=i|X_{n-\ell}=x_{n-\ell})}{P(X_{n+\ell}=x_{n+\ell}|X_{n-\ell}=x_{n-\ell})},\nonumber
\end{align}
for $i=0,1$. Let $P_k\triangleq\Pr(X_{n+k}\neq X_n)$. With this notation we have that if $x_{n+\ell}\neq x_{n-\ell}$ then
\begin{align}
P_1^\ell(x_{n+\ell},x_{n-\ell})=P_0^\ell(x_{n+\ell},x_{n-\ell})&=\frac{P_\ell(1-P_\ell)}{P_{2\ell}}.\label{eq:diffSample}
\end{align}
On the other hand, if $x_{n+\ell}= x_{n-\ell}$ we have
\begin{align}
P_1^\ell(x_{n+\ell},x_{n-\ell})P_0^\ell(x_{n+\ell},x_{n-\ell})&=\frac{P_\ell^2}{1-P_{2\ell}}\frac{(1-P_\ell)^2}{1-P_{2\ell}}.\label{eq:sameSample}
\end{align}
It therefore follows that
\begin{align}
\text{MMSE}\left(\ell\right)&=\Pr(X_{n+\ell}\neq X_{n-\ell})\left(\frac{P_\ell(1-P_\ell)}{P_{2\ell}}\right)^2\nonumber\\
&+\Pr(X_{n+\ell}= X_{n-\ell})\left(\frac{P_\ell(1-P_\ell)}{1-P_{2\ell}}\right)^2\nonumber\\
&=\left(P_\ell(1-P_\ell)\right)^2\left(\frac{1}{P_{2\ell}}+\frac{1}{1-P_{2\ell}}\right)\nonumber\\
&=\frac{\left(P_\ell(1-P_\ell)\right)^2}{P_{2\ell}(1-P_{2\ell})}.\label{eq:PtMMSE}
\end{align}
Note that
\begin{align}
P_k&=\Pr(X_{n+k}\neq X_n)\nonumber\\
&=\Pr\left(\left(\prod_{i=n+1}^{n+k}(-1)^{W_i}\right)=-1 \right)\nonumber\\
&=\frac{1-\Expt\left(\prod_{i=n+1}^{n+k}(-1)^{W_i}\right)}{2}\nonumber\\
&=\frac{1-\left(1-2q\right)^k}{2}\label{eq:PkExp}
\end{align}
Substituting~\eqref{eq:PkExp} into~\eqref{eq:PtMMSE} gives
\begin{align}
\text{MMSE}\left(\ell\right)&=\frac{\left(\frac{1}{4}\left(1-(1-2q)^{2\ell}\right)\right)^2}{\frac{1}{4}\left(1-(1-2q)^{2\ell}\right)\left(1+(1-2q)^{2\ell}\right)}\nonumber\\
&=\frac{1}{4}\cdot\frac{1-(1-2q)^{2\ell}}{1+(1-2q)^{2\ell}}.\label{eq:MMSEtExp}
\end{align}
Substituting~\eqref{eq:MMSEtExp} into~\eqref{eq:sumMMSE} gives
\begin{align}
\lim_{n\to\infty}4\frac{\mmse_{\bpi}(\bX)}{n}&\geq \sum_{t=0}^{\infty}2^{-(t+1)}\frac{1-(1-2q)^{2^{t+1}}}{1+(1-2q)^{2^{t+1}}}  \nonumber\\
&\geq\sum_{t=1}^{\infty}2^{-t}\frac{1-(1-2q)^{2^t}}{1+(1-2q)^{2^t}},\label{eq:mmseMarkov}
\end{align}
and consequently we get the following theorem.
\begin{theorem}
Let $\{X_n\}$ be a first-order Markov process with parameter $q$, $\{Z_n\}$ be an i.i.d. $\Ber(\alpha)$ process statistically independent of $\{X_n\}$ and $Y_n=X_n\oplus Z_n$. Then
\begin{align}
\overline{H}(Y)\geq h(\alpha)+\left(1-h(\alpha)\right)\sum_{t=1}^{\infty}2^{-t}\frac{1-(1-2q)^{2^t}}{1+(1-2q)^{2^t}}.\nonumber
\end{align}
\label{thm:MarkovLB}
\end{theorem}

\begin{remark}
For every $\alpha\in(0,1/2)$ there exist a $q_{\alpha}>0$ such that the bound from Theorem~\ref{thm:MarkovLB} outperforms Mrs. Gerber's Lemma for all $q\in(0,q_{\alpha})$. For example, $q_{0.11}\approx 0.212$.
As discussed in the previous section, $q_{\alpha}$ increases with $\alpha$ and approaches $1/2$ as $\alpha\to 1/2$.
\end{remark}

It will be instructive to study the behavior of the RHS of~\eqref{eq:mmseMarkov} in the limit of $q\rightarrow 0$. To this end we write, for some $0<\gamma<1$ such that $-\gamma\log{q}$ is an integer
\begin{align}
\lim_{n\to\infty}4&\frac{\mmse_{\bpi}(\bX)}{n}\geq \sum_{t=1}^{\infty}2^{-t}\frac{1-(1-2q)^{2^t}}{1+(1-2q)^{2^t}}\nonumber\\
&\geq \sum_{t=1}^{-\gamma\log{q}}2^{-t}\frac{1-(1-2q)^{2^t}}{2}\nonumber\\
&=\sum_{t=1}^{-\gamma\log{q}}2^{-(t+1)}\left(2^{t+1}q-\sum_{k=2}^{2^t}(-1)^k{{2^t}\choose k}(2q)^k \right)\nonumber\\
&\geq\sum_{t=1}^{-\gamma\log{q}}2^{-(t+1)}\left(2^{t+1}q-\sum_{k=2}^{2^t}(2^t)^k(2q)^k \right)\nonumber\\
&\geq\sum_{t=1}^{-\gamma\log{q}}q - 2^{-(t+1)}\sum_{k=2}^{2^t}\left(2^{t+1}q\right)^k. \label{eq:geometric}
\end{align}
Using the fact that $\sum_{k=2}^m r^k=\frac{r^2-r^{m+1}}{1-r}\leq\frac{r^2}{1-r}$ for $0<r<1$, we further bound~\eqref{eq:geometric} as
\begin{align}
\lim_{n\to\infty}4\frac{\mmse_{\bpi}(\bX)}{n}&\geq\sum_{t=1}^{-\gamma\log{q}}q - 2^{-(t+1)}\frac{\left(2^{t+1}q\right)^2}{1-2^{t+1}q} \nonumber\\
&=\sum_{t=1}^{-\gamma\log{q}}q - \frac{2^{t+1}q^2}{1-2^{t+1}q} \nonumber\\
&\geq-\gamma q\log{q}\left(1 - \frac{2 q^{1-\gamma}}{1-2 q^{1-\gamma}}\right). \label{eq:mmseboundGamma}
\end{align}
For $q\rightarrow 0$ we can take $\gamma=1-1/\sqrt{-\log{q}}$ such that
\begin{align}
\lim_{n\to\infty}4\frac{\mmse_{\bpi}(\bX)}{n}&\geq -q\log(q)\left(1-\varepsilon'_q\right)\nonumber\\
&=h(q)\left(1-\varepsilon_q\right)\label{eq:MarkovMMSElimit}
\end{align}
where $\varepsilon'_q,\varepsilon_q\rightarrow 0$ as $q\rightarrow 0$. We have therefore obtained that
\begin{align}
\lim_{q\to 0}\lim_{n\to \infty}4\frac{\mmse_{\bpi}(\bX)}{n h(q)}=\lim_{q\to 0}\lim_{n\to \infty} 4\frac{\mmse_{\bpi}(\bX)}{H(\bX)}\geq 1.\nonumber
\end{align}
Thus, we have seen that while the trivial choice $\bpi'=\{1,2.\ldots,n\}$ yields $\mmse_{\bpi'}(\bX)$ that meets the lower bound from~\eqref{eq:MMSEbounds}, the more clever choice of $\bpi$ given in~\eqref{eq:piMarkov} yields $\mmse_{\bpi}(\bX)$ that meets the upper bound from~\eqref{eq:MMSEbounds} in the limit.

\begin{remark}
The permutation $\bpi$ from~\eqref{eq:piMarkov} can be found by a greedy algorithm that constructs the permutation vector sequentially by choosing in the $i$th step
\begin{align}
\pi(i)=\argmax_{j\in[n]\setminus\{\pi(1),\ldots,\pi(i-1) \}}\text{MMSE}\left(X_j|X_{\pi(1)},\ldots,X_{\pi(i-1)}\right),\nonumber
\end{align}
where $[n]\triangleq\{1,\ldots,n\}$. The asymptotic optimality of $\bpi$ from~\eqref{eq:piMarkov} for symmetric Markov chains may suggest that such a greedy algorithm will always yield the permutation vector that maximizes $\mmse_{\bpi}(\bX)$. This is, unfortunately, not true in general. As a counterexample consider the vector $\bX=(X_1,X_2)$ with
\begin{align}
\Pr(X_1=0,X_2=0)=\frac{1}{2} \ &; \ \Pr(X_1=0,X_2=1)=0\nonumber\\
\Pr(X_1=1,X_2=0)=\varepsilon \ &; \ \Pr(X_1=1,X_2=1)=\frac{1}{2}-\varepsilon\nonumber
\end{align}
for which $\Var(X_1)>\Var(X_2)$ but
\begin{align}
\Var(X_2)+\text{MMSE}(X_1|X_2)>\Var(X_1)+\text{MMSE}(X_2|X_1)\nonumber
\end{align}
for $\epsilon$ small enough.
\end{remark}

Substituting~\eqref{eq:MarkovMMSElimit} into Theorem~\ref{thm:MarkovLB} gives that for small $q$
\begin{align}
\overline{H}(\bY)\geq h(\alpha)+(1-h(\alpha))h(q)(1-\varepsilon_q).\label{eq:MarkovEntRateLB}
\end{align}
Note that this bound has an infinite slope at $q=0$. This is always better than the Cover-Thomas type of bounds $\overline{H}(Y)\geq H(Y_m|Y_{m-1},\ldots,Y_1,X_0)$ derived in~\cite[Theorem 4.5.1]{Cover} which are always smaller than $h(q^{*m}*\alpha)$, where $q^{*m}$ denotes convolving $q$ with itself $m$ times. Both bounds evaluate to $h(\alpha)$ at  $q=0$, but the derivative of the latter is finite for any finite $m$. Thus, for small $q$ our bound is better than the Cover-Thomas bound of any order.

The bound~\eqref{eq:MarkovEntRateLB} is weaker than the best known lower bounds on $\overline{H}(Y)$ in the rare transition regime. For example, in~\cite{now05} it is shown that $\overline{H}(Y)\geq h(\alpha)-\tfrac{(1-2\alpha)^2}{1-\alpha}q\log {q}$, whereas in~\cite{pq11} this was improved to $\overline{H}(Y)\geq h(\alpha)+h(q)-Cq$ for some $C>0$. However, the two bounds mentioned above are ``tailor-made'' to hidden Markov models, whereas~\eqref{eq:MarkovEntRateLB} follows from applying our generic bound from Theorem~\ref{thm:MMSEGerber} to the special case of a hidden Markov model. In the next subsection we will show that the scalar version of our MMSE-based bound, stated in Lemma~\ref{lem:condMMSEGerber} can be used to enhance such a ``tailor-made'' bound for Markov chains.

\subsection{Bound based on the Ordentilch-Weissman Method}
In~\cite{ow11}, E. Ordentlich and T. Weissman cleverly observed that the entropy rate of a binary symmetric first-order hidden Markov process can be expressed as
\begin{align}
\overline{H}(\bY)=\Expt\left(\frac{e^{W_i}}{1+e^{W_i}}*q*\alpha \right),\label{eq:OWexpt}
\end{align}
where the auto-regressive process $W_i$ is defined as
\begin{align}
W_i=R_i\ln\frac{1-\alpha}{\alpha}+S_i f(W_{i-1})
\end{align}
for
\begin{align}
f(t)=\ln\frac{e^t(1-q)+q}{q e^t+(1-q)}
\end{align}
and i.i.d. processes $\{R_i\}$ and $\{S_i\}$ statistically independent of $W_0$, with distributions
\begin{align}
R_i=\begin{cases}
1 & w.p. \ 1-\alpha \\
-1 & w.p. \ \alpha
\end{cases} \ ; \
S_i=\begin{cases}
1 & w.p. \ 1-q \\
-1 & w.p. \ q
\end{cases}.
\end{align}
The expectation in~\eqref{eq:OWexpt} is taken under the assumption that $W_0$ is distributed according to the (unique) stationary distribution of the process $\{W_i\}$, and is therefore well-defined. In~\cite{ow11}, upper and lower bounds on $\overline{H}(\bY)$ were derived by analyzing the support of the process $\{W_i\}$. Here, we apply Lemma~\ref{lem:condMMSEGerber} in order to derive a lower bounds on $\overline{H}(\bY)$. To this end, we set $X|W_i\sim\Ber\left(\tfrac{e^{W_i}}{1+e^{W_i}} \right)$ and find a lower bounds on
\begin{align}
\mathrm{MMSE}(X|W_i)=\Expt\left(\frac{e^{W_i}}{(1+e^{W_i})^2}\right).\nonumber
\end{align}
Let $F\triangleq e^{f(W_{i-1})}$ and $\eta=\tfrac{1-\alpha}{\alpha}$, such that $e^{W_i}=\eta^{R_i} F^{S_i}$. We have
\begin{align}
\Expt\bigg(\frac{e^{W_i}}{(1+e^{W_i})^2}&|F\bigg)=(1-\alpha)(1-q)\frac{\eta F}{(1+\eta F)^2}\nonumber\\
&+(1-\alpha)q \frac{\eta/F}{(1+\eta/F)^2}\nonumber\\
&+\alpha(1-q)\frac{F/\eta}{(1+F/\eta)^2}+\alpha q\frac{(1/(\eta F)}{(1+1/(\eta F))^2}\nonumber\\
&=\left((1-\alpha)(1-q)+\alpha q\right) \frac{\eta F}{(1+\eta F)^2}\nonumber\\
&+\left((1-\alpha)q+\alpha(1-q)\right)\frac{F/\eta}{(1+F/\eta)^2}\label{eq:invsim}\\
&=(1-\alpha*q)\frac{\eta F}{(1+\eta F)^2}+(\alpha*q) \frac{F/\eta}{(1+F/\eta)^2}\nonumber\\
&\triangleq g(F),\label{eq:Fmmse}
\end{align}
where we have used the fact that $e^{W_i}/(1+e^{W_i})^2=e^{-W_i}/(1+e^{-W_i})^2$ in~\eqref{eq:invsim}. Let $\m{S}$ be the support of the random variable $F$. Clearly,
\begin{align}
\mathrm{MMSE}(X|W_i)=\Expt g(F)\geq \min_{s\in{S}} g(s)
\end{align}
In~\cite[eq. (44-45)]{ow11} it is shown that $\m{S}\subseteq[1/F_{\text{max}},F_{\text{max}}]$, where
\begin{align}
F_{\text{max}}\triangleq \frac{(\eta-1)(1-q)+\sqrt{4\eta q^2+(\eta-1)^2(1-q)^2}}{2\eta q}.\label{eq:Fmax}
\end{align}
Let $g_1(F)\triangleq \tfrac{\eta F}{(1+\eta F)^2}$ and $g_2(F)\triangleq \tfrac{F/\eta}{(1+F/\eta)^2}$, and note that $g_2(1/F)=g_1(F)$ and that $g(F)=(1-\alpha*q)g_1(F)+(\alpha*q)g_2(F)$. For $F\geq 1$ we have that $g_2(F)\geq g_1(F)$, whereas for $F<1$ we have that $g_1(F)>g_2(F)$. Since $(1-\alpha*q)\geq (\alpha*q)$ (recall that we assume $\alpha,q\leq 1/2$), we must have that
\begin{align}
\min_{s\in[1/F_{\text{max}},F_{\text{max}}]}g(s)=\min_{s\in[1,F_{\text{max}}]}g(s).
\end{align}
Straightforward algebra gives
\begin{align}
\sign&\left(g'(s)\right)\nonumber\\
&=\sign\left((\eta-s)(1+\eta s)^3-\frac{1-\alpha*q}{\alpha*q}(\eta s-1)(\eta+s)^3\right).\nonumber
\end{align}
Note that $\sign(g'(1))=-1$, and therefore if the equation $\sign\left(g'(s)\right)=0$ does not have any real solution in $[1,F_{\text{max}})$ then we must have
\begin{align}
\min_{s\in[1/F_{\text{max}},F_{\text{max}}]}g(s)=g(F_{\text{max}}).
\end{align}
Otherwise, $\min_{s\in[1/F_{\text{max}},F_{\text{max}}]}g(s)$ is obtained either in one of the solutions of $\sign\left(g'(s)\right)=0$ in the interval $[1,F_{\text{max}})$, or in $F_{\text{max}}$. The equation $\sign\left(g'(s)\right)=0$ is equivalent to
\begin{align}
&\eta\left(\frac{1-\alpha*q}{\alpha*q}+\eta^2\right)s^4+\left(3\eta^2\frac{1}{\alpha*q}-\eta^4-\eta\right)s^3\nonumber\\&+3\eta\frac{1-2(\alpha*q)}{\alpha*q}(\eta^2-1)s^2\nonumber\\
&+\left(\frac{1-\alpha*q}{\alpha*q}\eta^4+1-3\frac{\eta^2}{\alpha*q}\right)s-\eta\left(1+\frac{1-\alpha*q}{\alpha*q}\eta^2\right)=0,\label{eq:Roots}
\end{align}
Let $\m{S^*}$ be the set of solutions to the equation~\eqref{eq:Roots} in $[1,F_{\text{max}})$. We conclude that $\mathrm{MMSE}(X|W_i)\geq g(F^*)$ where
\begin{align}
F^*=\argmin_{s\in\left(\m{S}^*\cup F_{\text{max}}\right)} g(s).\label{eq:Fstar}
\end{align}
and this combined with~\eqref{eq:OWexpt} and Lemma~\ref{lem:condMMSEGerber} yields the following.
\begin{theorem}\label{thm:MarkovLB2}
Let $\{X_n\}$ be a first-order Markov process with parameter $q$, $\{Z_n\}$ be an i.i.d. $\Ber(\alpha)$ process statistically independent of $\{X_n\}$ and $Y_n=X_n\oplus Z_n$. Then
\begin{align}
\overline{H}(Y)\geq h(\alpha*q)+\left(1-h(\alpha*q)\right)g(F^*),\nonumber
\end{align}
where $F^*$ is defined by~\eqref{eq:Fmax},~\eqref{eq:Roots} and~\eqref{eq:Fstar}, $g(\cdot)$ is defined in~\eqref{eq:Fmmse}, and $\eta=\tfrac{1-\alpha}{\alpha}$.
\end{theorem}

In Figure~\ref{fig:BoundsComp} we plot the bound from Theorem~\ref{thm:MarkovLB2} for $\alpha=0.11$ and $q\in[0,0.5]$. For comparison, we also plot the lower bound from \cite[Corollary 4.8 and Lemma 4.10]{ow11}, and it is seen that for small values of $q$ our new bound improves upon that of~\cite{ow11}.

\begin{figure}[h]
\begin{center}
\includegraphics[width=1 \columnwidth]{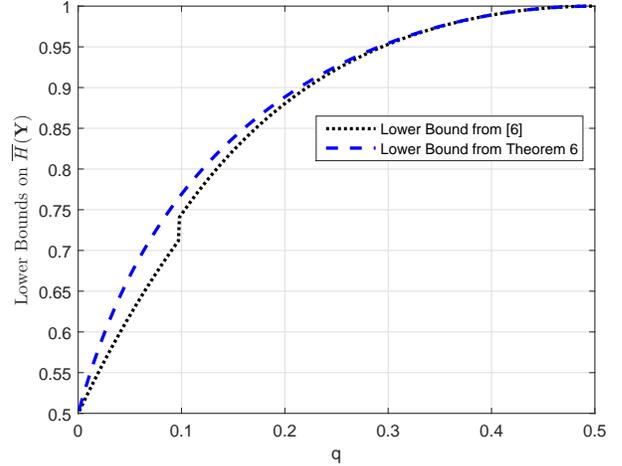}
\end{center}
\caption{Comparison between the lower bound from Theorem~\ref{thm:MarkovLB2} and the lower bound from \cite[Corollary 4.8 and Lemma 4.10]{ow11} for $\alpha=0.11$ and $q$ ranging between $0$ and $\tfrac{1}{2}$.}
\label{fig:BoundsComp}
\end{figure}

\bibliographystyle{IEEEtran}
\bibliography{OrBib2}




\end{document}